\documentclass[12pt, final]{l4dc2023}

\title[Learning-enhanced Nonlinear MPC with KNODE ensembles]{Learning-enhanced Nonlinear Model Predictive Control using Knowledge-based Neural Ordinary Differential Equations and Deep Ensembles}
\usepackage{times}

\author{%
    \Name{Kong Yao Chee} \Email{ckongyao@seas.upenn.edu}
 \newline \Name{M. Ani Hsieh} \Email{m.hsieh@seas.upenn.edu}
 \newline \Name{Nikolai Matni} \Email{nmatni@seas.upenn.edu}
 \newline \addr University of Pennsylvania, Philadelphia, PA 19104
}
\begin{document}
\maketitle
\vspace{-0.8cm}
\begin{abstract}
Nonlinear model predictive control (MPC) is a flexible and increasingly popular framework used to synthesize feedback control strategies that can satisfy both state and control input constraints. In this framework, an optimization problem, subjected to a set of dynamics constraints characterized by a nonlinear dynamics model, is solved at each time step. Despite its versatility, the performance of nonlinear MPC often depends on the accuracy of the dynamics model. In this work, we leverage deep learning tools, namely knowledge-based neural ordinary differential equations (KNODE) and deep ensembles, to improve the prediction accuracy of this model. In particular, we learn an ensemble of KNODE models, which we refer to as the KNODE ensemble, to obtain an accurate prediction of the true system dynamics.  This learned model is then integrated into a novel learning-enhanced nonlinear MPC framework. We provide sufficient conditions that guarantees asymptotic stability of the closed-loop system and show that these conditions can be implemented in practice. We show that the KNODE ensemble provides more accurate predictions and illustrate the efficacy and closed-loop performance of the proposed nonlinear MPC framework using two case studies.
\end{abstract}


\vspace{-0.4cm}
\section{Introduction} 
\vspace{-0.2cm}
Advances in nonlinear optimization algorithms and improvements in hardware computational power, nonlinear model predictive control (MPC) is proliferating across a range of autonomous systems applications, \textit{e.g.}, \cite{chowdhri2021integrated, yao2021state}. For similar reasons, deep learning tools are also becoming prevalent, particularly in the context of learning dynamical system representations. These tools can be used to extract meaningful information of the system from data, which is then incorporated into the dynamics models for MPC. Within the MPC framework, the models are deployed to predict the system dynamics, which consequently enables accurate closed-loop performance. Despite a surge in recent literature, the choice of deep learning tools and how they can be seamlessly integrated to learn dynamics models within the MPC framework remains a relatively open problem (\cite{mesbah2022fusion}). In this work, we present a novel learning-enhanced nonlinear MPC framework that uses a suite of deep learning techniques to enhance the accuracy of the dynamics model, which improves the performance of the closed-loop system.

\textbf{Related Works.}
Learning-based model predictive control (LMPC) is a fast-growing research area and a variety of topics have been studied in recent literature. Here, we highlight recent developments, particularly in the context of learning nonlinear dynamics models, and refer interested readers to recent reviews by \cite{hewing2020learning} and \cite{mesbah2022fusion}. There are a number of works, \textit{e.g.}, \cite{kabzan2019learning}, \cite{maddalena2020learning}, \cite{torrente2021data}, that adopt a non-parametric approach to learning dynamics models, before incorporating them into an MPC framework. However, it is well known that non-parametric models such as Gaussian processes may not be amenable to large amounts of data. Hence, a data selection strategy may be required to distill the relevant data for training and deployment. On the other hand, parametric models such as feedforward neural networks (NNs) and its variants, have also been explored within the LMPC context.  \cite{chen2020transfer} use NNs to model heating, ventilation and air-conditioning systems (HVAC) within smart buildings and integrate it into an MPC scheme. \cite{son2022learning} utilized NNs to learn the model-plant mismatch and applied it into an offset-free MPC scheme. Architectural variants such as recurrent neural networks and long short-term memory networks have also been considered in LMPC, \textit{e.g.}, \cite{terzi2021learning,wu2021statistical,saviolo2022physics}. These networks are structurally more complex and incur significant computational costs when integrated with a model predictive controller. In general, literature that apply ensemble methods in the context of LMPC seem to be sparser. \cite{kusiak2010modeling} use NN ensembles to model HVAC systems and use it in a particle swarm optimization problem. However, these processes operate at slower time scales and may not be suitable for fast-paced applications. A related concept is scenario-based MPC, \textit{e.g.}, \cite{bernardini2009scenario, schildbach2014scenario}, where scenarios are sampled from a given distribution and used in linear models within an MPC framework. One difference in this approach is that these scenarios collected from the system are in terms of the uncertainty and not of the states and control inputs. A second difference is in the formulation of the optimization problem, where the scenarios are considered simultaneously, instead of fusing them using a selected strategy, which is more common in ensemble methods.  For a recent review on ensemble methods in deep learning, we refer the readers to \cite{ganaie2021ensemble}. 

\textbf{Contributions.} Knowledge-based neural ordinary differential equations (KNODE), which comprises of a nominal first-principles model and a neural ordinary differential equation (NODE) component, are used to extract dynamics models which are then utilized within a MPC framework (\cite{chee2022knode}). In this work, we leverage the fact that KNODE models are relatively lightweight and introduce deep ensembling to enhance the learned dynamics model. We show that this provides more accurate state predictions and improves closed-loop performance. Our contributions in this work are three-fold. First, we present a novel learning-enhanced nonlinear MPC framework, leveraging concepts in deep learning, namely KNODEs and deep ensembles. These tools allow us to assimilate data into the learned model in a systematic and amenable manner, which in turn provides an accurate representation of the system dynamics and improves closed-loop performance. Second, for our proposed framework, we provide sufficient conditions that render the closed-loop system under the proposed MPC framework asymptotically stable. We show that these conditions, when applied with our proposed framework, are practical and can be implemented readily. Third, using two case studies, we demonstrate the efficacy and versatility of our proposed framework.

\vspace{-0.4cm}
\section{Problem Framework and Preliminaries} \label{section:problem_formulation}
\vspace{-0.2cm}
Consider a discrete-time, nonlinear system with the following dynamics,
\vspace{-0.2cm}
\begin{equation} \label{eq:dynamics}
    x^{+} = f(x,u),
\vspace{-0.2cm}
\end{equation}
where $x,x^+ \in \mathbb{R}^n$ are the current and successor states, $u \in \mathbb{R}^m$ is the control input to the system and it is assumed that $f$ is twice continuously differentiable. To control the system, we consider a nonlinear MPC framework that generates a sequence of optimal control inputs, which operates in a receding horizon manner. First, a dynamics model is constructed using first principles or prior knowledge of the system. Using this model, together with information on the state and control input constraints, we formulate the following finite horizon constrained optimal control problem,
\vspace{-0.15cm}
\begin{subequations} \label{eq:ftocp}
\small\begin{align}
    \;\underset{U}{\textnormal{minimize}}\quad\; & \sum_{i=0}^{N-1} q\left(x_i, u_i\right) + p(x_N)\\
    \text{subject to}\quad\; &x_{i+1} = \hat{f}(x_i, u_i), \quad \forall\, i=0,\dotsc,N-1,\label{eq:mpc_model_intro}\\
    \; &x_i \in \mathcal{X}, \quad u_i \in \mathcal{U}, \quad \forall\, i=0,\dotsc,N-1,\\
    \; &x_N \in \mathcal{X}_f,\quad x_0 = x(k), \label{eq:initial_state}
\end{align}
\end{subequations}
where $N$ is the horizon, $q(x_i,u_i)$ and $p(x_N)$ are the state and terminal state costs, and the sets $\mathcal{X},\,\mathcal{U}$ and $\mathcal{X}_f$ denote the constraint sets for the states, control inputs and the terminal state respectively. At each time step $k$, using state measurements $x(k)$, we solve \eqref{eq:ftocp} and obtain a sequence of optimal control inputs, $U^{\star} := \{u_0^{\star},\dotsc,u_{N-1}^{\star}\}$. We apply the first vector in this sequence $u_0^{\star}(x(k))$ to the system as the control action. Hence, the control law is expressed as $u(k) = u_0^{\star}(x(k))$ and the closed-loop system is given as $x^+ = f(x,u_0^{\star}(x)) := f_{cl}(x)$. 

If the dynamics model in \eqref{eq:mpc_model_intro} represents the system \eqref{eq:dynamics} with sufficient accuracy, we can expect the MPC framework to perform well. However, like many models, it is unlikely for the dynamics model to capture the true dynamics perfectly. To improve upon this dynamics model, we utilize deep learning tools to enhance its accuracy within the MPC framework. We first adopt a NODE formulation to represent any unknown, poorly modeled, and/or residual dynamics. During training, we combine this NODE model with a nominal model derived from prior knowledge to learn a {hybrid} KNODE model. This allows us to get a more accurate representation of the true dynamics in \eqref{eq:dynamics}. Next, instead of learning a single KNODE model, we learn an ensemble of KNODE models, which we denote as a KNODE ensemble. Through experiments, we show that the KNODE ensemble provides better accuracy and generalization in terms of state predictions and better closed-loop performance in terms of trajectory tracking.

\textbf{Notation.} $\mathbb{R}$ and $\mathbb{R}^n$ denote the set of real numbers and the $n$-dimensional vector space of real numbers. We use $||\cdot||$ to denote the Euclidean and spectral norm for vectors and matrices respectively. $B_{r}$ denotes a closed ball in $\mathbb{R}^n$ centered at the origin, \textit{i.e.}, $B_r := \{x \in \mathbb{R}^n \,|\, ||x|| \leq r \}$. For a matrix $A \in \mathbb{R}^{n \times n}$, $A \succeq 0$ and $A \succ 0$ imply that $A$ is positive semi-definite and positive definite respectively.

\vspace{-0.3cm}
\section{KNODE Ensembles} \label{section:knode}
\vspace{-0.1cm}
\subsection{Knowledge-based Neural Ordinary Differential Equations}
\vspace{-0.1cm}
For many electro-mechanical and robotic systems, it is relatively straightforward to obtain a model of the system, using first principles and prior knowledge. However, in the presence of uncertainty and modeling errors, it is challenging to assess if this derived model is sufficiently accurate, especially if it is applied within a model-based control framework such as nonlinear MPC. On the other hand, when the true system is operating in its specified environment, we can collect useful data that inform us about its underlying dynamics. The KNODE framework (\cite{jiahao2021knowledge}) provides a systematic approach to fuse the model obtained from first principles with the collected data. In this framework, we consider the following continuous-time representation of the true system,
\vspace{-0.1cm}
\begin{equation} \label{eq:cts_model}
\dot{{x}} = {f}_c(x,u) = \hat{{f}_c}({x},{u}) + d_c({x},{u}),
\end{equation}
where the subscript $c$ denotes the continuous-time nature and is used to distinguish these models from their discretized counterparts. The function $\hat{{f}_c}({x},{u})$ is a nominal model derived from prior knowledge. The dynamics of the true system ${f}_c(x,u)$ differs from that of the nominal model and this difference is represented by the unknown, residual dynamics, $d_c({x},{u})$. We parameterize these unknown dynamics using a NODE with parameters $\theta$, denoted by $d_{\theta,c}({x},{u})$. The model $\hat{f}_c(x,u) + d_{\theta,c}({x},{u})$ is jointly known as the KNODE model. To train this model, the first step is to collect a dataset $\mathcal{D}_{train}:=\left\{({x}_{i}, {u}_{i})\right\}_{i=1}^M$ that consists of $M$ data samples. Each of these samples contains the state measurement ${x}_i$ and the control input ${u}_i$, sampled at times $\{t_i\}_{i=1}^M$. For each data sample in $\mathcal{D}_{train}$ and with the nominal model $\hat{{f_c}}$, we generate one-step state predictions,
\vspace{-0.2cm}
\begin{equation} \label{eq:one_step}
    \Hat{{x}}_{{i+1}} = {x}_{i} +  \int^{t_{i+1}}_{t_i} \hat{f_c}({x}_{i},{u}_{i}) + d_{\theta,c}({x}_{i}, {u}_{i})\, dt,
\vspace{-0.3cm}
\end{equation}
where the integral is computed with a numerical method, \textit{e.g.}, the explicit $4^{\text{th}}$ order Runge-Kutta (RK4) method. Next, we define a loss function that characterizes the mean-squared errors across the predicted and measured states, 
\vspace{-0.5cm}
\begin{equation}
    \mathcal{L}({\theta}) := \frac{1}{M-1}\sum^{M}_{i=2}\left\|\hat{{x}}_{i} - {{x}}_{i}\right\|^2.
\vspace{-0.3cm}
\end{equation}

To obtain a set of optimal parameters ${\theta}^{\star}$, the model is put through a backpropagation procedure, where the gradients of the loss function with respect to the parameters are computed with automatic differentiation and applied in an optimization routine, Adam (\cite{Kingma2015AdamAM}). After training, we obtain a KNODE model, $\hat{{f_c}}({x},{u}) + d_{{\theta}^{\star},c}({x},{u})$. It has been shown that a single KNODE model provides predictions of higher accuracy as compared to the nominal model, as described in \cite{chee2022knode} and \cite{ jiahao2022online}. In this work, we propose an enhancement by learning an ensemble of KNODE models, instead of a single model.

\vspace{-0.2cm}
\subsection{Deep Ensembles} \label{section:weights}
\vspace{-0.1cm}
To improve upon the prediction accuracy and to reduce the variance of the predictions, we introduce a further enhancement to the KNODE model. Inspired by deep ensembling methods such as those described by \cite{lakshminarayanan2017simple} and \cite{wilson2020bayesian}, we construct a deep ensemble of KNODE models. It has been shown that ensembles provide improvements across a variety of classification and regression tasks in a number of application areas, \textit{e.g.}, see reviews in \cite{dietterich2000ensemble}, \cite{ganaie2021ensemble}. In our proposed framework, we adopt a randomized approach where each model in the KNODE ensemble is trained with the same training dataset, with the parameters in each model initialized at random. To promote diversity in the ensemble, we allow the models to have different network architectures. The ensemble residual dynamics is obtained by computing a weighted average across the individual models, \textit{i.e.}, $d_{ens,c} := \sum_{j=1}^{L} \alpha_j d_{\theta^{\star}_j,c}$, where $L$ is the number of models in the ensemble and $\{\alpha_i\}_{j=1}^L$ are weights for the models in the ensemble, with $\sum_j \alpha_j = 1$. This formulation can also be interpreted as a form of stacked regression (\cite{ganaie2021ensemble}). Two variants of the KNODE ensemble are proposed. In the first variant, equal weights are assigned to each model in the ensemble. For the second variant, we compute a set of weights for the models through an optimization problem. Specifically, we formulate and solve the following problem across a hold-out dataset, $\mathcal{D}_{v} := \{(x_i,u_i)\}_{i=1}^{M_{v}}$,
\vspace{-0.2cm}
\begin{subequations} \label{eq:nnls}
\small
\begin{align}
    \underset{\alpha_1,\dotsc,\alpha_L}{\textnormal{minimize}}\quad\; &\frac{1}{M_{v}-1}\sum^{M_{v}}_{i=2} || x_i - \sum^{L}_{j=1} \alpha_j \Hat{x}_{i,j}||^2\\
    \text{subject to}\quad\; & \Hat{{x}}_{i+1,j} = {x}_{i} +  \int^{t_{i+1}}_{t_{i}} \hat{f_c}({x}_{i},{u}_{i}) + d_{{\theta_j^{\star},c}}({x}_{i}, {u}_{i})\, dt,\;\; \forall i=1,\dotsc,M_{v}-1, \label{eq:nnls_integral}\\
    \quad\; \vspace{-0.2cm} &\sum_{j=1}^L \alpha_j = 1. \label{eq:norm_constraint}
\end{align}
\vspace{-0.5cm}
\end{subequations}

\noindent where $\Hat{x}_{i,j}$ denotes the state predicted by the $j^{\text{th}}$ model in the ensemble for the $i^{\text{th}}$ data sample in $\mathcal{D}_{v}$. Similar to \eqref{eq:one_step}, the integral in \eqref{eq:nnls_integral} is computed using a numerical solver. The hold-out set $\mathcal{D}_v$ is obtained by splitting the collected dataset into $\mathcal{D}_{train}$ and $\mathcal{D}_v$. 
\vspace{-0.2cm}
\section{Learning-based Nonlinear MPC with KNODE Ensembles} \label{section:knode_ensemble_mpc}
\vspace{-0.1cm}
We integrate the trained KNODE ensemble into a learning-enhanced nonlinear MPC framework. In particular, we replace the dynamics model in \eqref{eq:mpc_model_intro} with a discretized version of the KNODE ensemble. The resulting finite horizon optimal control problem is written as 
\vspace{-0.2cm}
\begin{subequations} \label{eq:ftocp_knode_ensemble}
\small
\begin{align}
    J^{\star}(x(k)) \;= \;\underset{U}{\textnormal{minimize}}\quad\; & \sum_{i=0}^{N-1} q\left(x_i, u_i\right) + p(x_N) \label{eq:mpc_cost}\\
    \text{subject to}\quad\; &x_{i+1} = \hat{f}(x_i, u_i) + d_{ens}^{\star}(x_i,u_i) , \quad \forall\, i=0,\dotsc,N-1, \label{eq:knode_mpc_model}\\
    \; &x_i \in \mathcal{X}, \quad u_i \in \mathcal{U}, \quad \forall\, i=0,\dotsc,N-1,\\
    \; &x_N \in \mathcal{X}_f, x_0 = x(k),
\end{align}
\end{subequations}
where $\hat{f}(x,u) + d_{ens}^{\star}(x,u) := \hat{f}(x,u) + \sum_{j} \alpha^{\star}_j \,d_{\theta^{\star}_j}(x,u)$ denotes the discretized version of the trained KNODE ensemble with optimal parameters $\theta^{\star}$ and $\alpha_j^{\star}$, $j\in\{1,\dotsc,L\}$. The control law is synthesized using a similar approach as described in Section \ref{section:problem_formulation}. Upon solving \eqref{eq:ftocp_knode_ensemble}, we obtain a sequence of optimal control inputs, $U^{\star} := \{u_0^{\star},\dotsc,u_{N-1}^{\star}\}$. The first element in this optimal control sequence $u_0^{\star}(x(k))$ is applied to the system as the control action and this proceeds in a receding horizon manner. Next, we examine the stability properties of the learning-based MPC framework by showing that there exists a practical choice of $q(x,u)$, $p(x)$ and $\mathcal{X}_f$ that guarantees asymptotic stability for the closed-loop system.

\vspace{-0.2cm}
\subsection{Conditions for Asymptotic Stability}
\vspace{-0.2cm}
For a compact set $\mathcal{X} \subset \mathbb{R}^n$ containing the origin, a function $V : \mathbb{R}^n \rightarrow \mathbb{R}_+$ is a Lyapunov function in $\mathcal{X} \subseteq \mathbb{R}^n$ for the system $f_{cl}(x)$ if 
\vspace{-0.2cm}
\begin{equation} \label{eq:Lyapunov_function}
\small
\begin{split}
    &V(0) = 0,\quad V(x) > 0,\quad \forall x \in \mathcal{X} \setminus \{0\},\\
    & V(f_{cl}(x)) - V(x) \leq -\alpha(||x||),\quad \forall x \in \mathcal{X} \setminus \{0\},
\end{split}
\end{equation}
where $\alpha : \mathbb{R}^n \rightarrow \mathbb{R}$ is a continuous, positive definite function. It is well known that the existence of a Lyapunov function implies that the origin is asymptotically stable for the system in $\mathcal{X}$ (\cite{borrelli2017predictive, chellaboina2008nonlinear}). In this work, we show that there exists a practical choice of the stage and terminal costs $q(x,u),\,p(x)$, and terminal state constraint set $\mathcal{X}_f$ such that the closed-loop system, under the proposed learning-based MPC scheme, is asymptotically stable. Before describing this design choice, we first state a set of conditions, conditions (a) through (e) in the following theorem, such that if they are fulfilled, then the value function in \eqref{eq:ftocp_knode_ensemble}, $J^{\star}(x(k))$, is a valid Lyapunov function and consequently, asymptotic stability is achieved. After which, we relate this design choice with these conditions. The following is a standard result that guarantees asymptotic stability for the closed-loop system.

\vspace{-0.2cm}
\begin{theorem} [\cite{borrelli2017predictive, grune2017nonlinear}] \label{theorem1}
\noindent Suppose that

$\;$ a) $\,$ q(x,u) and p(x) are continuous and positive definite functions,

$\;$ b) $\,$The set $\mathcal{X}_f \subseteq \mathcal{X}$ is closed and contains the origin in its interior,

$\;$ c) $\,$ $v(x) \in \mathcal{U},\quad\forall x \in \mathcal{X}_f$,

$\;$ d) $\,$ $f(x,\,v(x)) \in \mathcal{X}_f,\quad\forall x \in \mathcal{X}_f$,

$\;$ e) $\,$ $p(f(x,\,v(x))) - p(x) + q(x,\,v(x)) \leq 0,\quad\forall x \in \mathcal{X}_f$,

\noindent where the functions $q(x,u),\,p(x),\,f(x,u)$, sets $\mathcal{X},\,\mathcal{U},\,\mathcal{X}_f$ are as described in \eqref{eq:ftocp_knode_ensemble} and $v(x)$ is a chosen local control law. Then, the value function $J^{\star}(x(k))$ satisfies the conditions in \eqref{eq:Lyapunov_function} and the origin of the closed-loop system $f_{cl}(x)$ is asymptotically stable in $\mathcal{X}_0$, where $\mathcal{X}_0$ is the set of initial states $x(k)$ for which the optimization problem \eqref{eq:ftocp_knode_ensemble} is feasible. 
\end{theorem}
\vspace{-0.1cm}
The proof of Theorem \ref{theorem1} is listed in the Appendix for completeness. Quadratic cost functions, \textit{i.e.}, $q(x,u) := x^{\top}Qx + u^{\top}Ru$, and $p(x) := x^{\top}Px$, are commonly used in the formulation of \eqref{eq:ftocp_knode_ensemble}. If $Q,\,R$ and $P \succ 0$, condition (a) is satisfied. While there is some flexibility in the choice of $Q$ and $R$, we show that a particular choice of terminal cost matrix $P$, terminal state constraint set $\mathcal{X}_f$, and local control law $v(x)$ ensure that conditions (b) through (e) are fulfilled. The underlying idea is to utilize the fact that the dynamics of a nonlinear system and its linearized dynamics are not too different in some neighborhood of the origin. This concept is not new and has been studied in the form of receding horizon control, synonymous to MPC, for nonlinear constrained systems (see \cite{chen1998quasi, magni2001stabilizing, kwon2005receding, rawlings2017model}). However, there are nuances in some of these approaches, which may make them prohibitive for implementation. For instance, in \cite{magni2001stabilizing}, the terminal cost is a function of the parameters characterizing the exponential stability of the linearized system. In \cite{chen1998quasi}, continuous-time systems are considered. In this work, we consider discrete-time nonlinear systems and provide an alternative perspective of the conditions in the Lyapunov equation, differing from those in \cite{kwon2005receding} (cf. (7.23) and (7.28) therein). In Section \ref{section:case_studies}, we further show that this choice of $P$ and $\mathcal{X}_f$, together with our proposed framework, can be implemented and performs well in practice. 

We consider a linearization of the system in \eqref{eq:dynamics} about the origin. Assuming that $f(0,\,0)=0$, the linear dynamics can be written as
\begin{equation} \label{eq:linear_dynamics}
\small x^{+} = \nabla f_x(x,u) \big|_{(x,u)=(0,0)} + \nabla f_u(x,u) \big|_{(x,u)=(0,0)}:= Ax + Bu,
\end{equation}
where the pair $(A,B)$ is assumed to be stabilizable. Next, we choose a linear feedback controller, $u_{l}(x) := Kx$, such that the linear closed-loop system, $x^+ = Ax + Bu_{l}(x) = (A+BK)x := A_{cl} x$, is exponentially stable. It is important to note that while this linear gain $K$ is chosen and required for the computation of the terminal cost matrix $P$, it is not required for the control law in MPC. In particular, by considering the local control law $v(x)$ in Theorem \ref{theorem1} to be $u_{l}(x)$, and together with $q(x,u)$, $p(x)$ and $\mathcal{X}_f$ as described below, it can be shown that condition (e) in Theorem \ref{theorem1} is fulfilled. This is given by the following proposition.

\begin{proposition} \label{proposition1}
Consider the stage and terminal costs to be $q(x,u) = x^{\top} Q x + u^{\top} R u$, $p(x) := x^{\top}Px$, where $Q,\,R\succ 0$, and $P$ is the solution to the following Lyapunov equation,
\begin{equation} \label{eq:Lyap_eqn}
    A_{cl}^{\top} P A_{cl} + \rho (Q + K^{\top} R K) - P = 0,
\end{equation}
where $\rho \in (1,\infty)$. Let the terminal constraint set $\mathcal{X}_f$ be a $\gamma$-sublevel set of $p(x)$, \textit{i.e.}, $\mathcal{X}_f := \{x\in\mathbb{R}^n \,|\, p(x) \leq \gamma\}$, with $\gamma > 0$. 
Then, it holds that
\begin{equation} \label{eq:nonlinear_descent}
    p(f(x,\,Kx)) - p(x) \leq -x^{\top}(Q + K^{\top} R K) x,\quad\forall x \in \mathcal{X}_f
\end{equation}
\end{proposition}
To prove Proposition \ref{proposition1}, we make use of a lemma that provides an upper bound on the ratio of the norm of a twice continuously differentiable function and the norm of its argument. This lemma with its proof, together with the proof of Proposition \ref{proposition1}, are given in the Appendix.

\vspace{-0.2cm}
\begin{corollary} \label{corollary}
Given $p(x)$ and $\mathcal{X}_f$ defined in Proposition \ref{proposition1}, together with the descent condition in \eqref{eq:nonlinear_descent}, conditions (b) and (d) in Theorem \ref{theorem1} are satisfied. 
\end{corollary}

\vspace{-0.3cm}
The proof for Corollary \ref{corollary} is given in the Appendix. For condition (c) in Theorem \ref{theorem1}, it is shown in \cite{magni2001stabilizing} that there exists a $\delta_1 \in (0,\infty)$ such that for any $x \in \mathcal{X}$, $Kx \in \mathcal{U}$, for all $x \in B_{\delta_1}$. Hence, to satisfy all the conditions in Theorem \ref{theorem1}, we consider an additional criterion to the selection of $\epsilon$ in the proof of Proposition \ref{proposition1} such that $\epsilon \in (0,\min\{\delta,\,\delta_1\}]$ and choose the $\gamma$-sublevel set accordingly. Through the case studies in Section \ref{section:case_studies}, we show that this choice of $\mathcal{P}$ and $\mathcal{X}_f$ given in Proposition \ref{proposition1} can be readily implemented, which guarantees asymptotic stability of the closed loop system under the proposed learning-enhanced MPC framework. We note that the conditions given above are not only valid for the proposed learning-based MPC scheme, but they can also be applied to other suitably formulated nonlinear MPC schemes. This approach for the selection of the terminal cost and constraint set may be more conservative than some recent results, \textit{e.g.}, \cite{kohler2019nonlinear, eyubouglu2022snmpc}. Examining and comparing with these results is a direction for future work. 

\vspace{-0.3cm}
\subsection{Enabling Efficient Implementation}
\vspace{-0.2cm}
Due to the nonlinear nominal dynamics and KNODE ensemble, \eqref{eq:ftocp_knode_ensemble} is a nonlinear constrained optimization problem. To ensure that it is solved in an efficient manner, we incorporate a number of measures in our implementation. We formulate \eqref{eq:ftocp_knode_ensemble} in CasADi (\cite{Andersson2019}) as a parametric optimization problem. This allows the solver to retain a fixed problem structure during deployment. A solver based on the interior-point method, IPOPT (\cite{wachter2006implementation}), is used to solve the optimization problem. It is warm-started at each time step by providing it with an initial guess, based on the solution from the previous time step. Since NODEs are used to learn only the residual and not the full dynamics, the trained neural networks are relatively lightweight in both width and depth. This further reduces the computational efforts required to solve \eqref{eq:ftocp_knode_ensemble}. While the evaluation of the models in the KNODE ensemble are computationally cheaper as compared to training them, there is a trade-off between the number of models in the ensemble and the required computational efforts. In Section \ref{section:case_studies}, we show that a reasonably small number of models improves the prediction accuracy of the dynamics model.

\vspace{-0.3cm}
\section{Case Studies} \label{section:case_studies}
\vspace{-0.2cm}
We consider two case studies to illustrate the flexibility and efficacy of the proposed learning-enhanced nonlinear MPC framework.

\vspace{-0.4cm}
\subsection{Inverted Pendulum}
\vspace{-0.2cm}
We first consider an inverted pendulum system with dynamics given by
(\cite{brockman2016openai}),
\vspace{-0.2cm}
\begin{equation} \label{eq:pendulum_eqns}
    \ddot{\theta} = 3g\sin(\theta)/(2l) + 3\tau/(ml^2),
\vspace{-0.2cm}
\end{equation}
where $\theta$ is the angle made by the pendulum and the vertical, $m$ and $l$ are the mass and length of the pendulum, $g$ is the gravitational force, and $\tau$ is the external torque acting on the pendulum. We define the state and control input to be $x :=[\theta\; \dot{\theta}]^{\top}$ and $u:=\tau$. To ascertain the learning ability of the KNODE ensemble, residual dynamics are introduced in the form of a mass difference. In particular, we consider the mass in the nominal model to be 1kg, while the mass of the true pendulum is 0.55kg. The inverted pendulum is required to track a time-varying reference trajectory defined by a sequence of step angle commands. The dynamics are simulated using an explicit RK45 method with a sampling time of 0.01s. To get the terminal cost matrix $P$, we linearize and discretize \eqref{eq:pendulum_eqns}, and compute the matrix $K$ that optimizes the discrete-time linear-quadratic regulator for the chosen values of $Q$ and $R$. We then solve the discrete-time Lyapunov equation as described in Proposition \ref{proposition1} to get $P$ and subsequently, $\mathcal{X}_f$. For this system, we set the parameter $\gamma$ to $0.01$. More details on training, the ensemble architecture and controller parameters are provided in the Appendix.

\vspace{-0.35cm}
\subsection{Quadrotor System}
\vspace{-0.2cm}
Next, we consider a high-dimensional quadrotor system. The dynamics of the system are given as (\cite{mellinger2011minimum})
\vspace{-0.3cm}
\begin{equation}\label{eq:eom}
    m\ddot{{r}} = m{g} + {R} \eta,\quad
    {I}\dot{{\omega}} = {\tau} - {\omega} \times {I} {\omega},
\vspace{-0.2cm}
\end{equation}
where ${r},\,{\omega} \in \mathbb{R}^3$ are the position and angular rate vectors of the quadrotor, $\eta \in \mathbb{R}$, ${\tau} \in \mathbb{R}^3$ are the thrust and moments generated by the motors of the quadrotor. The gravity vector is denoted by ${g}$ and ${R}$ is the transformation matrix that maps $\eta$ to the acceleration of the quadrotor. $m$ and ${I}\in \mathbb{R}^{3\times3}$ are the mass and inertia matrix of the quadrotor. The state and control inputs are defined as ${x} := [{r}^{\top}\;\dot{{r}}^{\top}\;{q}^{\top}\;{\omega}^{\top}]^{\top} \in \mathbb{R}^{13}$ and ${u} := [\eta\; {\tau}^{\top}]^{\top} \in \mathbb{R}^4$, where ${q} \in \mathbb{R}^4$ denotes the quaternions that represents the orientation of the quadrotor. The quadrotor is tasked to track circular trajectories of varying radii and speeds, which implies tracking time-varying position and velocity commands. It is assumed that there is an unknown disturbance force that is a function of the velocities acting on the quadrotor. The cost matrix $P$ and terminal constraint set $\mathcal{X}_f$ are obtained using a similar approach as the first case study. Details on training, the ensemble architecture and controller parameters are listed in the Appendix.

\vspace{-0.4cm}
\section{Results and Discussion} \label{section:results}
\vspace{-0.1cm}
\subsection{Predictions by the KNODE Ensemble}
\vspace{-0.1cm}
For the inverted pendulum, we evaluate the prediction accuracy of the KNODE ensemble by computing the statistics of the mean-squared errors (MSE) across 30 test trajectories. For each of these trajectories, both the initial states and the magnitude of the step commands are selected at random, under an uniform distribution. Statistics of the prediction accuracy are shown in Fig. \ref{fig:stats_pen}. We compare the MSEs across three cases; (i) individual KNODE models, (ii) the first variant, \textit{i.e.}, a KNODE ensemble of equal weights, and (iii) the second variant, \textit{i.e.}, a KNODE ensemble with different weights. Observed from the left panel of Fig. \ref{fig:stats_pen}, the median errors for the two variants are 18.1\% and 56.2\% lower than those given by the predictions across the individual models. 

Next, we analyze the results by comparing the MSEs for each of the test runs, as depicted in the right panel of Fig. \ref{fig:stats_pen}. The MSEs for the average model are computed by taking the average across the models in the ensemble. As observed, the MSEs for the second variant are lower than those for the first variant and the average model. These improvements indicate that the KNODE ensemble provides more accurate predictions and generalizes better to unseen trajectories, since these test runs are different from those in the training data. It is observed from the right panel of Fig. \ref{fig:stats_pen} that the first variant provides lesser MSE reduction, as compared to the second variant. These results illustrate the effectiveness of the weight optimization procedure described in Section \ref{section:weights}. 

For the quadrotor system, we compute MSEs of the states across 30 trajectories, where the commanded speed and radii, and initial states are selected at random. Statistics of the MSEs are given in Fig. \ref{fig:stats_quad}. An improvement is observed for the KNODE ensemble with different weights, as compared to the other two schemes, in terms of the median errors. We further examine the MSEs for each of the 30 runs, as shown in the right panel of Fig. \ref{fig:stats_quad}. While it is observed that both variants of the KNODE ensemble provide improved accuracy and lower MSEs for all runs over the average model, the application of different weights provides MSE reduction in 60\% of the runs, when compared against using equal weights in the ensemble. This can be attributed to the challenge of finding suitable weights for the ensemble as the required complexity of the models increases, particularly for the quadrotor. 
\setlength{\textfloatsep}{-1pt}
\begin{figure}
    \centering
    {\includegraphics[scale=0.4, trim = 0cm 1cm 0cm 0cm]{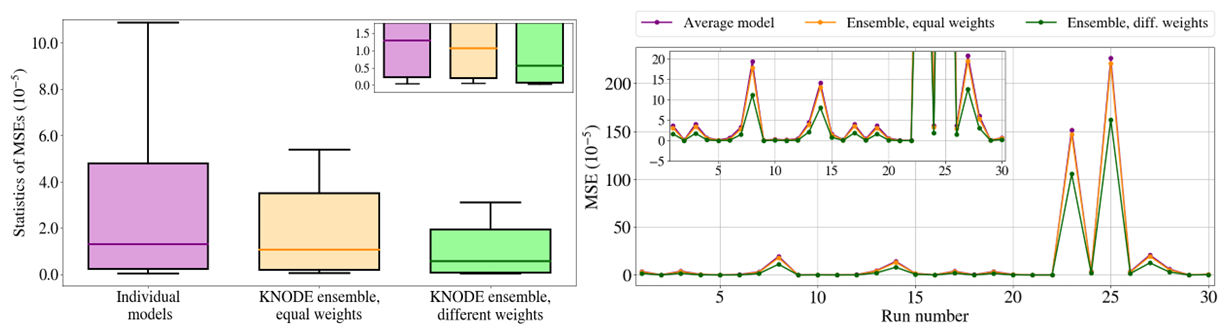}}
    \caption{MSEs for individual KNODE models, KNODE ensembles with equal and different weights across 30 test trajectories, for the inverted pendulum. The left panel depicts the error statistics and the right panel shows the MSEs for each test trajectory. 
    }
    \label{fig:stats_pen}
\end{figure}

\begin{figure}
    \centering
    {\includegraphics[scale=0.4, trim = 0cm 2cm 0cm 1.5cm]{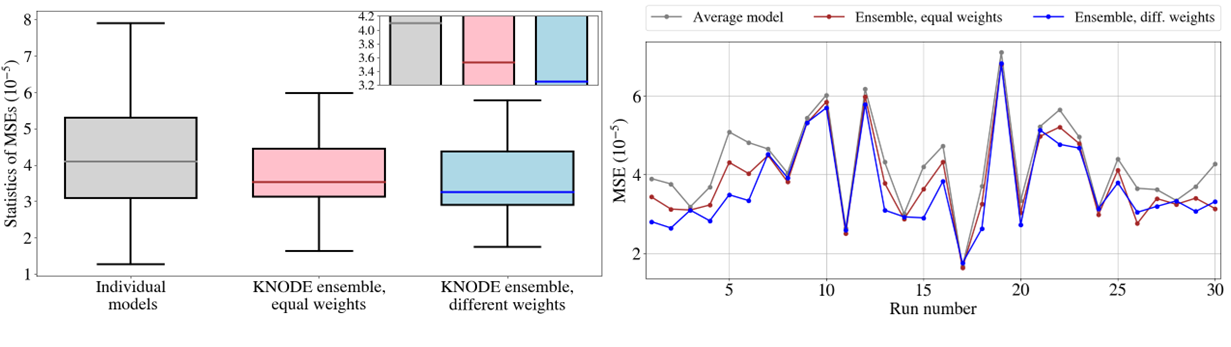}}
    \caption{MSEs for individual KNODE models, KNODE ensembles with equal and different weights across 30 test trajectories, for the quadrotor. The left panel depicts the error statistics and the right panel shows the MSEs for each test trajectory. 
    }
    \label{fig:stats_quad}
\end{figure}

\vspace{-0.3cm}
\subsection{Closed-loop Tracking Performance}
\vspace{-0.2cm}
To evaluate the performance of the inverted pendulum, we consider the time histories of the tracking errors between the reference and closed-loop angle trajectories. We compare the errors under three schemes; learning-enhanced MPC with (i) individual KNODE models, (ii) a KNODE ensemble of equal weights and (iii) a KNODE ensemble with different weights. These schemes are denoted as \emph{KNODE-MPC}, \emph{EnKNODE-MPC, Equal} and \emph{Different}. Fig. \ref{fig:cl_plots} depicts the time histories of the angle tracking errors across test trajectories. It is observed that using \emph{KNODE-MPC} results in higher variance in the steady-state errors, as compared to the two ensemble approaches. Moreover, the median steady-state errors are smaller in magnitude for the MPC schemes that use a KNODE ensemble. This implies that using a KNODE ensemble leads to more consistent closed-loop performance. On the other hand, the number of models required for the ensemble would depend on the dynamics, model complexity and the amount of computational power available. While significant prediction improvements are obtained with the KNODE ensemble with different weights, as shown in Fig. \ref{fig:stats_pen}, we only observe a slight improvement in the error bands for \emph{EnKNODE-MPC, Opt}, as compared to \emph{EnKNODE-MPC, Equal}. This implies that for this inverted pendulum, improvements in model accuracy may not necessarily lead to significant improvements in the closed-loop performance.

For the quadrotor, the task involves following a circular trajectory, which translates to the tracking of time-varying position and velocity commands. The time histories of the position errors for one test case are shown in the left panel of Fig. \ref{fig:cl_quad}. The errors are computed by taking the difference between the quadrotor and reference trajectories. The reference trajectories in the x and y axes are time-shifted for ease of comparison, as there is a time difference between the reference and quadrotor trajectories, due to the time-varying nature of the reference trajectories. As observed, the errors are slightly smaller with the KNODE ensemble in the MPC scheme, as compared to a single model. Next, we evaluate the overall closed-loop performance through the MSE statistics for the positions and velocities across 30 runs and they are shown in the right panel of Fig. \ref{fig:cl_quad}. The position MSEs under \emph{EnKNODE-MPC, Equal} and \emph{EnKNODE-MPC, Different} are lower than \emph{KNODE-MPC} by 3.8\% and 1.7\% and the velocity MSEs are lower by 24.1\% and 10.4\% respectively, illustrating the efficacy of the framework.  
\begin{figure}
    \centering
    {\includegraphics[width=9cm, height=3cm, scale=0.5, trim = 0.5cm 1cm 0cm 0.5cm]{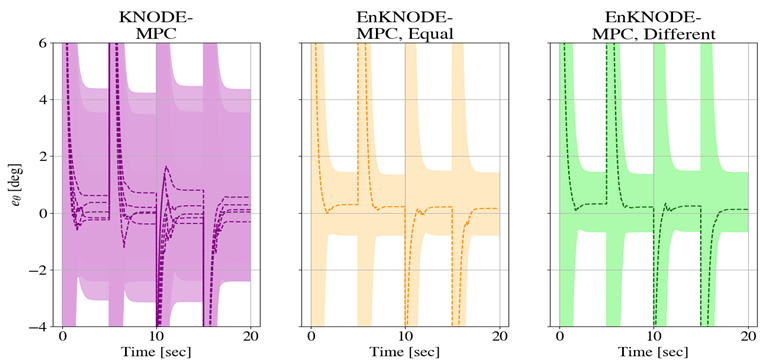}}
    \caption{Time histories of the errors between the reference and angle trajectories for the pendulum, under the three learning-based MPC schemes. The dashed lines represent the time histories of the median across trajectories, while the colored bands indicate the maximum and minimum tracking errors across test runs.}
    \label{fig:cl_plots}
\end{figure}

\begin{figure}
    \centering
    {\includegraphics[scale=0.4, trim = 0.5cm 1cm 0cm 1.5cm]{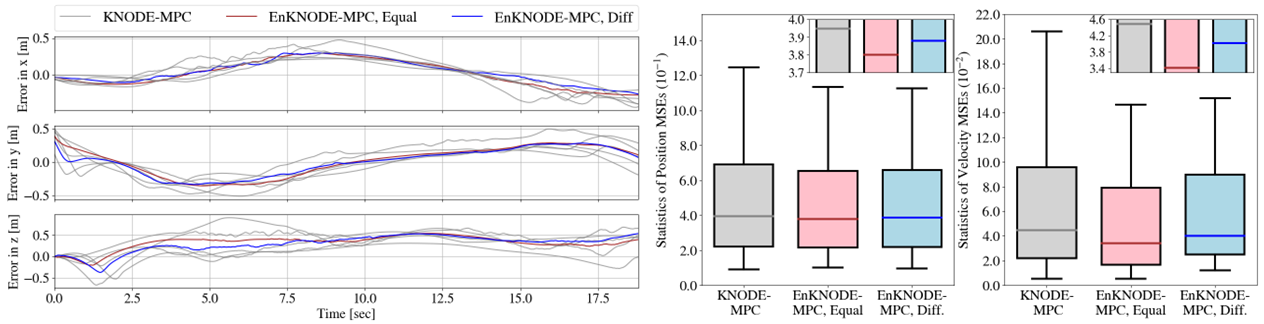}}
    \caption{\textbf{Left}: Time histories of the position errors for the quadrotor for a single test case. For this test case, the MSEs for EnKNODE-MPC, Equal and Diff. are 1.125, while the MSEs for KNODE-MPC with the single models range from 1.128 to 1.260. \textbf{Right}: Statistics of the quadrotor position and velocity tracking MSEs under the three MPC schemes. The insets are zoomed-in plots of the median errors.}
    \label{fig:cl_quad}
\end{figure}

\section{Conclusion}
We present a learning-enhanced nonlinear MPC framework that incorporates concepts in deep learning, namely KNODEs and deep ensembles. For the proposed framework, sufficient conditions that render the closed-loop system asymptotically stable are provided. Through two case studies, we demonstrate that the proposed scheme not only provides more accurate predictions, but also improves the closed-loop performance. For future work, we aim to further reduce the computational efforts required, and improve performance for the framework, by combining ideas from multi-stage nonlinear MPC.


\newpage
\acks{This work is supported by DSO National Laboratories, 12 Science Park Drive, Singapore 118225, NSF IUCRC 1939132, and the Office of Naval Research (ONR) Award No. N00014-22-1-2157.}

\bibliography{root.bib}

\newpage
\section{Appendix}
\subsection{Proof of Theorem \ref{theorem1}}
\begin{proof}
The first part of the proof is analogous to the proof for Theorem 12.2 in \cite{borrelli2017predictive}, which can be obtained by replacing the linear dynamics therein with the nonlinear dynamics of \eqref{eq:knode_mpc_model}. We note that in the proof of this theorem, the upper bound for $J^{\star}(x)$ holds for all $x \in \mathcal{X}_f$. To show that this upper bound is valid for all $x \in \mathcal{X}_0$, we consider Propositions 2.15 and 2.16 in \cite{rawlings2017model}. In particular, Proposition 2.15 states that $J^{\star}(x)$ is locally bounded, \textit{i.e.}, bounded on every compact subset of $\mathcal{X}_0$. Following Proposition 2.16, it can be shown that this local boundedness property of $J^{\star}(x)$, together with the fact that $J^{\star}(x)$ is continuous at the origin, implying the upper bound holds for all $x \in \mathcal{X}_0$.  
\end{proof}

\subsection{Upper bounding the ratio between the norm of a function and its argument}
\begin{lemma} \label{sup_bound_lemma}
Consider a twice continuously differentiable function $e:\mathbb{R}^n \rightarrow \mathbb{R}^n$ and suppose that $||e(0)|| = \left|\left|\nabla e(x)\big|_{x=0}\right|\right| = 0$. For any $\delta > 0$, there exists $e_{\delta} > 0$ such that
\begin{equation}
    \frac{||e(x)||}{||x||} \leq \frac{1}{2} e_{\delta} ||x||, \quad\forall x \in B_{\delta},
\end{equation}
where $B_{\delta} := \{x \in \mathbb{R}^n \,|\,||x|| \leq \delta\}$. 
\end{lemma}
\begin{proof}
Since $e(x)$ is twice continuously differentiable and $B_{\delta}$ is closed and bounded, for any $\delta > 0$, there exists $e_{\delta} > 0$ such that $||\nabla^2 e(x)|| \leq e_{\delta},\;\forall x \in B_{\delta}$. Next, we consider two arbitrary points $x_1,\,x_2 \in B_{\delta}$ and apply a mean value version of the Taylor's theorem, derived from the fundamental theorem of calculus and integration by parts, to get
\begin{equation}
    e(x_1) = e(x_2) + \nabla_x e(x)\Big|_{x=x_2} (x_1-x_2) + \int_0^1 (1-t) (x_1-x_2)^{\top} \nabla^2 e(x)\Big|_{x=x_2+t(x_1-x_2)} (x_1-x_2) dt.
\end{equation}
Setting $x_2 = 0$ and taking its norm gives
\begin{equation}
\begin{split}
||e(x_1)||&=\left|\left|e(0) + \nabla e(x)\Big|_{x=0} x_1 + \int_0^1 (1-t) x^{\top} \nabla^2 e(x)\Big|_{x=tx_1} x_1 \,dt\right|\right|\\
&\leq \left|\left| \int_0^1 (1-t) x_1^{\top} \nabla^2 e(x)\Big|_{x=tx_1} x_1 \,dt \right|\right|\\
&\leq \int_0^1 (1-t) \left|\left| \nabla^2 e(x) \Big|_{x=tx_1} \right|\right| ||x_1||^2\,dt\\
&\leq \frac{1}{2} e_{\delta}||x_1||^2,\quad\forall x_1 \in B_{\delta}.
\end{split}
\end{equation}
Re-arrangement gives the required result.
\end{proof}

\subsection{Proof of Proposition \ref{proposition1}}
\begin{proof}
(Adapted from \cite{rawlings2017model}). First, if $Q,\,R \succ 0$, then $Q + K^{\top} R K \succ 0$, and since $A_{cl}$ is stable, $P \succ 0$. For the linear system \eqref{eq:linear_dynamics}, since $p(x) := x^{\top}Px$, we have
\begin{equation} \label{eq:linear_descent}
    p(A_{cl} x) - p(x) = -\rho x^{\top} (Q+K^{\top} R K) x, \quad\forall x \in \mathbb{R}^n.
\end{equation}
Considering \eqref{eq:nonlinear_descent} and \eqref{eq:linear_descent}, it remains to show that
\begin{equation} \label{eq:mixed_descent}
    p(f(x,\,Kx)) - p(A_{cl} x) \leq (\rho-1)x^{\top}(Q + K^{\top} R K) x,\quad\forall x \in \mathcal{X}_f.
\end{equation}
To this end, we introduce an error term to account for the difference between the nonlinear and linear systems under the linear feedback control law, $e(x) := f(x, Kx) - A_{cl} x$. The left-hand side of \eqref{eq:mixed_descent} can then be re-written as
\begin{equation}
\begin{split}
    p(f(x,\,Kx)) - p(A_{cl} x) &= 2x^{\top}A_{cl}^{\top} P e(x) + e(x)^{\top} P e(x)\\
    &\leq 2||x||||e||||A_{cl}^{\top}P||+||P||||e||^2\\
    &\leq \left(2S_{\delta}||A_{cl}^{\top}P|| + S_{\delta}^2 ||P|| \right)||x||^2,
\end{split}
\end{equation}
where $S_{\delta} := \sup_{x\in B_{\delta}} ||e(x)||/||x||$, as given in Lemma \ref{sup_bound_lemma}. From Lemma \ref{sup_bound_lemma}, we also have $\lim_{{\delta} \rightarrow 0} S_{\delta} =0$. This implies that there exists an $\epsilon \in (0, \delta]$ such that 
\begin{equation} \label{eq:norm_inequality}
    \left(2S_{\delta}||A_{cl}^{\top}P|| + S_{\delta}^2 ||P|| \right)||x||^2 \leq (\rho-1)x^{\top}(Q + K^{\top} R K) x. 
\end{equation}
To ensure that \eqref{eq:norm_inequality} holds for all $x \in \mathcal{X}_f$, we select the terminal state constraint set $\mathcal{X}_f$, in particular the parameter for the sublevel set $\gamma$, such that $||x|| \leq \epsilon$. Given the choice of the stage and terminal costs and $P$ in \eqref{eq:Lyap_eqn}, there exists a $c > 0$ such that
$
    c ||x||^2 \leq q(x,u)\leq p(x) \leq \gamma, \;\forall x \in \mathcal{X}_f.
$
Hence, by choosing $\gamma := c\epsilon^2$, we have $||x|| \leq \epsilon$. For this choice of $\gamma$, \eqref{eq:norm_inequality} holds for all $x \in \mathcal{X}_f$, which gives the required result.
\end{proof}

\subsection{Proof of Corollary \ref{corollary}}
\begin{proof}
The terminal constraint set $\mathcal{X}_f = \{x\in\mathbb{R}^n \,|\, p(x) \leq \gamma\}$ is closed and contains the origin in its interior, by construction. Since $Q,R \succ 0$, the descent condition gives $p(f(x,\,Kx)) \leq p(x) \leq \gamma,\, \forall x \in \mathcal{X}_f$.
This implies that $f(x,\,Kx) \in \mathcal{X}_f,\,\forall x \in \mathcal{X}_f$ and condition (d) is fulfilled.
\end{proof}

\subsection{Training of KNODE Ensemble}
For the inverted pendulum, the training dataset is collected using nominal MPC, and it consists of 2000 data samples and spans 20 seconds. 75\% of these data samples are used to train the KNODE ensemble and 25\% are used as the hold-out set to optimize the weights in the ensemble. For the ensemble, we trained 5 networks. Each network consists of the input and outer layers, and with 1 hidden layer and the hyperbolic tangent activation function. The number of neurons in this hidden layer varies from 64 to 320, with increments of 64 neurons for each model. Training is done in PyTorch (\cite{NEURIPS2019_9015}). Each of these models is trained for 500 epochs, and a learning rate of 2e-2 and a weight decay of 1e-8. Adam (\cite{Kingma2015AdamAM}) is used as the optimizer. We use the Python-based torchdiffeq library (\cite{chen2018neuralode}) for the numerical integration during training, and the explicit Runge Kutta 4$^{\text{th}}$ order method (RK4) within the library was used. PyTorch and Adam are also utilized for the ensemble weight optimization procedure. The learning rate and weight decay are set at 2e-3 and 1e-9 respectively.

For the quadrotor, the training dataset is collected using a nominal model predictive controller. It consists of 4000 data samples and spans 20 seconds. 75\% of these data samples are used to train the KNODE ensemble and 25\% are used to optimize the weights on the models in the ensemble. For the ensemble, we trained 5 networks. Each network consists of the input and outer layers, with one hidden layer and the hyperbolic tangent as the activation function. The number of neurons in this hidden layer varies from 8 to 40, with the addition of 8 neurons for each model. Training is also performed in PyTorch. Each of these models is trained for 1000 epochs, and we use a learning rate of 2e-2 and a weight decay of 1e-9. Similarly, we use Adam as the optimizer and use torchdiffeq and RK4 for the numerical integration during training. PyTorch and Adam are also utilized for the weight optimization procedure. It is optimized for 1500 epochs and the learning rate and weight decay are set at 3e-2 and 1e-9 respectively.

\subsection{Controller Parameters}
For the inverted pendulum, the cost matrices are chosen to be $Q :=\text{diag}[1.0,\,0.1],\,R :=$1e-5, and the control horizon is set to 10. The parameters $\rho$ and $\gamma$ are set to be 1.1 and 0.01 respectively. 
For the quadrotor, the cost matrices are set as \[Q :=\text{diag}[0.05,\,0.05,\,0.05,\,0.05,\,0.05,\,0.05,\,0.1,\,0.1,\,0.1,\,0.1,\,0.1,\,0.1,\,0.1],\]
\[R :=\text{diag}[0.001,\,0.001,\,0.001,\,0.001],\]
and the control horizon is set to 20. we set $\rho$ and $\gamma$ to be 1.1 and 0.5 respectively. We further impose control constraints on the commanded thrust to be between $[0,\,0.575]$. 

\end{document}